\documentclass{IET}%%%%where IET is the template name
\usepackage[lined,algonl,algoruled]{algorithm2e} %,algochapter
\usepackage{algorithmic}
\usepackage{setspace}
\usepackage{placeins}
\usepackage{amssymb}

\newcommand{\tuple}[1]{\langle #1 \rangle}
\newcommand{\remove}[1]{}

\newtheorem{theorem}{Theorem}
\newtheorem{condition}{Condition}
\newtheorem{lemma}{Lemma}

\begin{document}

\title{Energy Efficient Restoring of Barrier Coverage in Wireless Sensor Networks Using Limited Mobility Sensors}

\author{Dinesh Dash}
\affil{NIT Patna, Bihar, India}
\author[2]{Anurag Dasgupta}
\affil{Valdosta State University, Valdosta, GA 31698, USA}

%\author[2,*]{ Anurag Dasgupta}
%%%% By default, the citations will come automatically,
%%%% The optional bracket "[2.*]" is used  to display the corresponding author symbol

\abstract{In Wireless Sensor Networks, sensors are used for tracking objects, monitoring health and  observing a region/territory for different environmental parameters. Coverage problem in sensor network ensures quality of monitoring a given region. Depending on  applications different measures of coverage are there. Barrier coverage is a type of coverage, which ensures all paths that cross the boundary of a region intersect at least one sensor's sensing region. The goal of the sensors is to detect intruders as they cross the boundary or as they penetrate a protected area. The sensors are dependent on their battery life. Restoring barrier coverage on sensor failure using mobile sensors with minimum total displacement is the primary objective of this paper. A centralized barrier coverage restoring scheme is proposed to increase the robustness of the network. We formulate restoring barrier coverage as  bipartite matching problem.  A distributed restoring of  barrier coverage algorithm is also proposed, which restores it by first finding existing alternate barrier. If alternate barrier is not found, an alternate barrier is reconstructed by shifting existing sensors in a cascaded manner. Detailed simulation results are shown to evaluate the performance of our algorithms.

}

\maketitle

\section{Introduction}
%% Different measures of coverage

Coverage problem in sensor network measures the quality of surveillance provided by the deployed sensors. Various definitions are popular to measure the quality of coverage such as {\em target coverage} \cite{cardei:2005}, {\em line coverage} \cite{agnetis:2009}, {\em area coverage} \cite{huang:2005}, {\em barrier coverage} \cite{kumar:2005} etc. {\em Target coverage} demonstrates that a set of target points lie within the sensing range of some sensors. {\em Line coverage} ensures that a given set of line segments must be fully covered by the sensors. {\em Area coverage} guarantees all the points within the region are covered by the sensors.

%% Basic definition of barrier coverage

{\em Barrier coverage} ensures that the sensor network surrounds the boundary of an area with sensors such that all the paths that cross the boundary must go through the sensing range of some sensors. If the sensor network ensures that all crossing paths intersect at least $k$ sensors' sensing regions then the boundary of the area is $k$-barrier covered. Ensuring {\em barrier coverage} is a challenging issue in sensor network for intruder detection in continental border. Nowadays, deploying mobile sensor networks is common and it is extremely useful in hostile environments such as battlefields and hazardous areas. In case of random sensor deployment, Saipulla et al. \cite{saipulla:2010}, showed that usage of the number of mobile sensors with limited mobility to achieve barrier coverage is significantly lesser than the number of static sensors. Recently, extensive research has been going on different measures for barrier coverage, scheduling scheme to prolong the lifetime of the network and deployment schemes to ensure barrier coverage both using static sensors and mobile sensors.

%% coverage maintenance scheme and motivation of our work 

Due to non deterministic initial deployment there may exist uncovered regions or for power drainage sensors may fail and create uncovered regions. Therefore, recovery schemes are required to ensure the quality of coverage and it is a challenging issue in sensor network. In literature \cite{yu:2008,zhang:2005,manoj:2007}, different dynamic area coverage schemes are proposed by using the existing mobile sensors to recover the coverage hole.  Most of the works on restoring coverage are focused on area coverage but there is no existing work for restoring barrier coverage. In this paper, we are focused on restoring barrier coverage. Initially, a barrier is formed by a subset of the deployed sensors. The sensors on the barrier are treated as {\em barrier nodes} and the remaining sensors are treated as {\em non-barrier nodes}. As time passes, {\em barrier nodes} may fail and create uncovered passages through which an intruder can cross the barrier. Restoring  barrier coverage is required to ensure quality of coverage. In \cite{wang:2014}, authors proposed a  barrier coverage restoration scheme. They estimate the lower bound of the number of mobile sensors to interconnect any two static sensor nodes and a progressive scheme to interconnect two static sensors. But using this scheme repairing a broken barrier is difficult because it assumes that someone is supplying mobile sensors on fly as needed. On the other side, we have proposed a quite easy and practical solution to repair a broken barrier with the help of nearby redundant mobile sensors with minimum total displacement. For practicability and scalability, we have also propose a distributed barrier coverage restoration algorithm (DBCRA). Our simulations results show that DBCRA outperforms the randomized scheme in both in terms of total displacement and successful restoring of barrier.

%% Contributation

In this paper, our main contributions are summarized as follows:

\begin{itemize}
   \item We propose a barrier coverage restoring scheme after sensor node failure using cascading node shifting. The algorithm ensures  the maximum displacement capacity constraint of the sensors based on their remaining energy and the  total displacement made by all the sensors  is minimum.
   \item We formulate the relocating of mobile sensors to repair the barrier as a minimum cost  maximum bipartite matching problem, and solve it in polynomial time using the Hungarian algorithm. 
   \item  A distributed barrier restoring algorithm is also proposed.
   \item Performance of the algorithm is evaluated and compared.  
\end{itemize}

%% paper organization

The rest of the paper is organized as follows. Section 2 briefly discusses related work on coverage maintenance schemes. Section 3 presents some necessary background and defines our problem. In Section 4, our centralized algorithm is presented for reconstructing the barrier on failure of one or multiple sensors and their time complexity is analyzed. A distributed algorithm is discussed in Section 5. Section 6 presents the experimental results. Finally, Section 7 concludes the paper.

\section{Related Works}
\label{related_work}

%%  coverage maintenance scheme

In \cite{manoj:2007}, Sekhar et al. propose a dynamic coverage maintenance (DCM) scheme to improve the area coverage of a mobile sensor network during its lifetime by exploiting the inherent redundancy in coverage. Four new heuristics are proposed for selection of the migrating sensors. In \cite{yu:2008}, authors propose a scheme where redundant sensors on the boundary of an uncovered area are determined in a distributed manner. Here each sensor independently finds whether it is a redundant sensor in a probabilistic way which depends on the number of neighbor. Similarly, when a sensor finds itself as a boundary node of an uncovered area, it issues a control packet over the network for calling other redundant sensors. On receiving such control packet by a redundant sensor it first registers itself to the initiator and then moves to the uncovered area. Zhang et al. \cite{zhang:2005} present a distributed self healing area coverage scheme where few sensors are mobile and the remaining are static. It works in three steps. In the first phase, it determines the boundary of the uncovered region. In the second phase, locations of the mobile sensors to cover the uncover region is determined. Finally, it calls the mobile sensors to place them in the target positions.

%%  Different measures of barrier coverage and scheduling schemes to ensure barrier coverage

Kumar et al. \cite{kumar:2005} have defined the barrier coverage problem
and propose a centralized algorithm to verify barrier coverage. In \cite{kumar:2007}, they have proposed a centralized scheduling scheme to achieve maximum life time by switching off the redundant barrier for both homogeneous and heterogeneous battery life time sensors.

The algorithm uses minimum number of path switching. In \cite{liu:2008}, Liu et al. have proposed a distributed scheduling algorithm to provide {\em strong barrier} coverage with high probability and low communication cost. They prove that if the sensor density reaches a certain value in a rectangular area of width $w$ and length $l$ with $w =\Omega(\log l)$, there exist multiple disjoint barriers with high probability. Chen et al. \cite{chen:2007} have proposed a measure for barrier coverage, known as $L$-local barrier coverage using a distributed approach. Finally, a distributed scheduling scheme is proposed to maximize the life time of the network. This scheme is able to verify global barrier coverage with high probability for thin barrier. They extend it for arbitrary belt and for heterogeneous sensors.

%% Relocation schemes to achieve barrier coverage

Shen et al. have proposed both centralized and distributed algorithm in \cite{shen:2008} to solve a constrained version of a barrier coverage problem by using mobile sensor where all the sensors are mobile. There is a given initial deployment of a set of sensors. Sensors relocate themselves on a line with equal distances between them such that they form a barrier and the total movement energy consumption is minimum. In \cite{bhattacharya:2009}, Bhattacharya et al. have shown that the mobile sensors have detected the existence of an unknown crossing path. The sensors reposition themselves most efficiently within a specified region such that they repair the existing security hole and thereby prevent intruders. Sensors are placed on to the perimeter of the region $R$ in equal separation such that (i) the longest movement is minimum and (ii) the total movement of the sensors is minimum. Saipulla et al. \cite{saipulla:2010}, have proposed an efficient sensor relocation algorithm so that after an initial deployment the sensors form maximum number of disjoint barriers with minimum movement. Here they assume that the sensors are deployed in a grid. The barriers are horizontal line with given positions. They proposed algorithm to verify whether there is a possibility to form barrier using the mobile sensors. Czyzowicz et al. \cite{czyzowicz:2010}, have proposed a sensors relocation strategy to cover a given line segment $I$ of length $L$. Initially, the sensors are deployed at arbitrary position on the line segment $I$. Sensors are moved such that the line is covered maximally and the total movement is minimum. They propose a centralized algorithm for sensors with equal sensing range. They have shown that for unequal sensing range the problem is NP-Complete. In \cite{wang:2013}, Wang et al. proposed a centralized reconstruction scheme to form a barrier from an initial deployment using directional sensor network using minimum number of mobile sensors with minimum total displacement. Wang et al. \cite{wang:2014} studied barrier coverage problem assuming that the nodes have location error. They established the minimum number of mobile sensors requirement to connect any to static sensors for two cases: only static sensors have location error and both static and mobile sensors have location error. They proposed a progressive deployment scheme using minimum mobile sensors to connect any two static nodes. 

%% Comparing with existing works

Compared to the previous works, our proposed work reconstructs an alternate barrier by finding alternate path or shifting nearby sensors to reconstruct the old barrier after node failure. To the best of our knowledge, this is the first work to study the barrier restoration  problem with limited mobility constraint sensors. Our scheme not only minimizes the total displacement of the sensors, it also considers the remaining energy constraint of the individual sensor to prolong the network life time. We proposed both centralized and distributed solution of the problem.

\section{Background and Problem Statement}
\label{sec:background}

% Background
Let $S=\{ s_1,s_2  \ldots s_N \}$ denote a set of sensors placed within a rectangular barrier region of length $L$ and width $W$. Let $P= \{p_1,p_2 \ldots p_N \}$ be the positions of the sensors. We assume that the sensors have uniform circular sensing range $\rho$ and the communication range of the sensors is at least twice the sensing range. The intersection graph of the sensing regions of the sensors is modelled as a {\em unit disk graph} \cite{ kumar:2010}. In the {\em unit disk graph}, each sensor $s_i$ is modelled as a vertex, and an edge is added between two vertices $s_i$ and $s_j$ in the graph if the sensing regions of $s_i$ and $s_j$ intersect. Two dummy vertices $p_L$ and $p_R$ are placed at the left and right ends of the rectangular boundary respectively. An edge is added between vertices $s_i$ and $p_L$ if the sensing region of $s_i$ intersects the left boundary of the rectangle. Similarly, an edge is added between vertices $s_j$ and $p_R$ if the sensing region of $s_j$ intersects the right boundary of the rectangle. Given the intersection graph, a barrier is formed by the sensors represented by the vertices on any path between $p_L$ and $p_R$ in the graph. Note that since there can be multiple such paths, multiple such barriers may exist. Any one such path can be chosen as the barrier, and the nodes on the path are called {\em barrier nodes} and remaining nodes are called {\em  non-barrier nodes}. An example of sensors network and its barrier is shown in Figure \ref{fig:graphical_representation}. Figure \ref{fig:graphical_representation}(a) shows the sensing region of the sensors and Figure \ref{fig:graphical_representation}(b) shows its corresponding intersection graph. The subset of sensors $\{s_1,s_3,s_5,s_6,s_8,s_9,s_{10},s_{12}\}$ are on a path between $p_L$ and $p_R$ in the intersection graph and hence they form a barrier.

\begin{figure}[t]
\centering
\includegraphics[width=14cm]{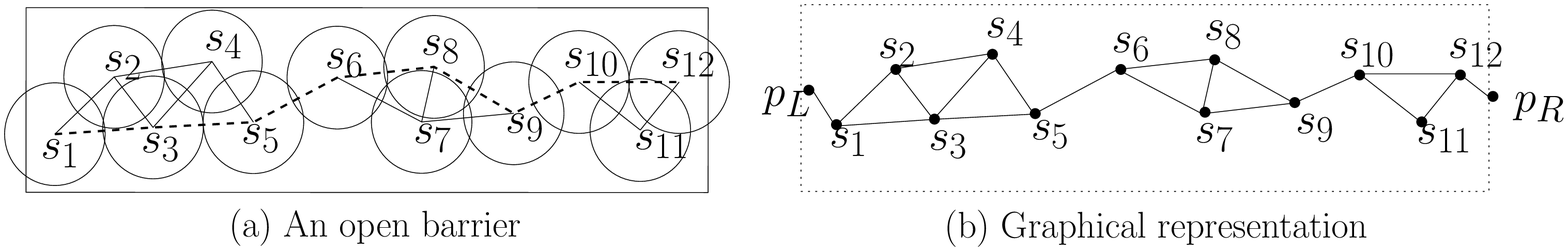}
\caption{A sensor network  and its intersection graph}
\label{fig:graphical_representation} 
\end{figure}

% General Problem Statement
Given an initial deployment of sensors such that there exists a path on the corresponding intersection graph between $p_L$ and $p_R$, an initial barrier coverage is achieved. However, due to node failure, this path may break and the sensor network may fail to ensure barrier coverage anymore. For example, in Figure \ref{fig:graphical_representation}(a), if sensor $s_6$ fails, there is no longer any path in the graph in Figure \ref{fig:graphical_representation}(b) between $p_L$ and $p_R$ and the sensor network fails to ensure barrier coverage. In this case, one possible way to reestablish barrier coverage is to move one or more sensors to new positions to recreate a path between $p_L$ and $p_R$. However, if $s_3$ fails, even though the initial barrier constructed $\{s_1, s_3, s_5, s_6, s_8, s_9, s_{10}, s_{12} \}$  is broken, the network can still ensure barrier coverage by using an alternate path between $s_1$ and $s_5$  (for ex. $\{ s_1, s_2, s_4, s_5, s_6, s_8, s_9, s_{10}, s_{12} \}$ ). In this case, no sensor needs to be moved; however, some sensors which were not on the barrier before the failure are now included in the barrier, and hence, marked as barrier nodes. Thus, on the failure of one or more sensors, barrier coverage can be restored by following one of two approaches, finding alternate path between appropriate sensors without moving any sensor, or by moving existing sensors to new positions. Since moving sensors is expensive in terms of energy, our main objective is  to minimize the total displacement made / energy consumed by the sensors as well as satisfy the remaining  displacement capacity / remaining energy constraints of each individual sensor.   

However, as discussed before, if sensor fails, the hole created in the barrier should ideally be repaired locally if possible by finding alternate path segments to breach the hole without changing the rest of the barrier farther away from the hole. If alternate path not found, then sensors need to be moved. Ideally, only nearby sensors should move. In this paper, we are solving a restricted version of the above general restoring barrier coverage problem, where  sensors are not allowed to move any arbitrary position but they are relocating themselves to one of the existing sensors positions within their limited displacement capacity, which may vary depending on their remaining energy. In the next two sections, we propose restoring barrier coverage  protocol: a centralized version as well as a simple distributed algorithm of it.

\section{Restoring Barrier Coverage Algorithm}
\label{sec:algorithm}

Our proposed algorithm for restoring barrier coverage works in two phases, an initialization phase and a maintenance phase. The initialization phase is executed once at the beginning to create/initialize a barrier and set up some data structures. The maintenance phase is invoked whenever node failure occurs. 

\subsection{Initialization Phase}

In the initialization phase, breadth first search technique is used to find a minimum hop count path between $p_L$ and $p_R$ in the intersection graph. All nodes on the path found are marked as {\em barrier nodes}. In Figure \ref{fig:graphical_representation}(a), the nodes on the  path marked by the dotted line are the {\em barrier nodes}. The remaining nodes are treated as {\em non-barrier nodes}. The status of a non-barrier node can change to barrier node if the barrier is modified after node failure and the non-barrier node participates to recreate the barrier.

\subsection{Maintenance Phase}

In the maintenance phase, on failure of a barrier node, it is first checked if an alternate barrier exists or not. 
In particular, we are searching existence of path between the predecessor and the successor of the failed node $s_f$. 
The alternate path is determined using  graph search method discussed in the following subsection. 
If no such path is found, then the algorithm attempts to replace the failed node by moving few sensors to new locations. The scheme for moving the sensors to reestablish the barrier is discussed next to alternate path determination subsection. Note that no maintenance needs to be done if a non-barrier  node fails, as the barrier is not affected by the failure.

\subsubsection{Determination of Alternate Path}
\label{sec:alternate_path}

In this sub-phase, we use breadth first search technique to find a minimum hop distance alternate path between the predecessor $s_{pre}$ and the successor $s_{suc}$ of a failed node $s_f$ in the intersection graph. Thereafter, merge the discovered path with the old barrier to find modified barrier.

\subsubsection{Shifting Sensors to Reestablish Barrier}
\label{sec:sensor_move}

If no alternate path is found between $s_{pre}$ and $s_{suc}$ in the intersection graph, the barrier cannot be repaired by our alternate path phase. This can be due to non-existence of alternate path between $s_{pre}$ and $s_{suc}$ in the intersection graph after the failure. In this case, the barrier coverage maintenance algorithm tries to shift sensors to new positions to reestablish the barrier. The neighboring barrier sensors should be moved such that the existing barrier is not compromised. While there can be many different approaches for shifting the sensors, in this paper we propose minimum displacement bipartite matching with cascaded node shifting algorithm ({\em MDBMCNS}) to repair the barrier. It will ensure maximum displacement capacity constraint of individual sensor as well as total displacement made by the sensors together is minimum. 

In {\em MDBMCNS} scheme, the active sensors together forms a set of vertices called sensor vertices represented by $L=\{s_1,s_2, \ldots s_N\} \setminus s_f $ whereas the positions of sensors on the barrier including the failed sensor forms another set of vertices called position vertices represented by $R=\{p_{i_i},p_{i_2}, \ldots p_{i_k} \}$ of the bipartite graph. There is an edge between a vertex $s_i \in L$ to a vertex $p_j \in R$ if $s_i$ has enough remaining energy to relocate itself to the position $p_j$. Once the existing network modeled as a bipartite graph, all edges of the form $(s_i, p_j)$ is assigned weights equal to euclidean distance between sensor $s_i$'s current position $p_i$ to $p_j$. Once the weight assignment of the bipartite graph is over, minimum cost based maximum matching of the graph is determined using Hungarian algorithm \cite{kuhn:1955}. After finding the matched edges between $L$ and $R$, if the number of matched edges is less than $|R|$ ( means number of active sensors are not sufficient to position themselves to the location of existing barrier sensors including failed sensors) then recovery is not possible. Otherwise, all the sensors incident with those matched edges are relocated to the corresponding positions to reform the old barrier. An example of sensor network of Figure \ref{fig:graphical_representation} after the failure of barrier sensor $s_6$ and its corresponding  bipartite graph are shown in Figure \ref{fig:bipartite}. In this example, $L=\{s_1,s_2, \ldots s_{12}\} \setminus s_6 $  and $R=\{ p_1,p_3, p_5, p_6, p_8, p_9, p_{10}, p_{12} \}$ and we assume that existing sensors have enough energy to replace any barrier neighbors. Otherwise, the corresponding edges are removed from the bipartite graph. For example, if sensor $s_8$ in Figure \ref{fig:bipartite}(b) does not have enough energy to relocate to positions $p_6$ and $p_9$ then the edges $(s_8,p_6)$ and $(s_8,p_9)$ must be dropped from the bipartite graph to ensure the limited energy constraint of the sensor. There is no edge between $s_1$ and $p_5$ because they are not in their communication range. This modelling not only fills a vacant barrier node position by a non-barrier neighbor but also allows cascaded movement of the barrier nodes until a non-barrier neighbor node is found.

%%  as well as $p_j$ is within the communication  range of $s_i$

\begin{figure}
\centering
\includegraphics[width=12cm]{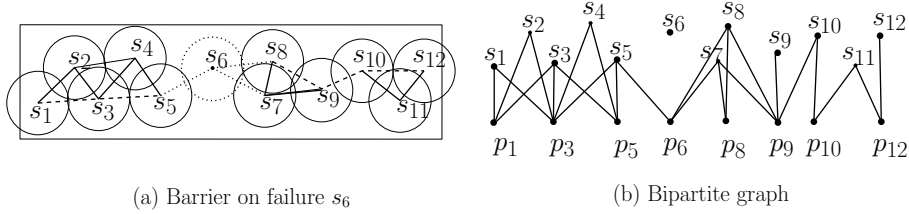}
\caption{Sensor network and its bipartite graph after node $s_6$ failure} \label{fig:bipartite}
\end{figure}

\subsection{Complexity Analysis}
\label{sec:analysis}
In this section, we establish the time complexity of different phases of the above protocol.

\begin{theorem} \label{th:1}
Initialization phase of barrier maintenance algorithm runs in $O(E)$ time, where $E$ is the number of edges in the intersection graph.
\end{theorem}

\begin{proof}
In this phase, breadth first search technique is used to find a path between $p_L$ and $p_R$ and establish the barrier. It runs in $O(E)$ time.
\end{proof}

%%%%%%%%%%%%%%%%%%%%%%%%%%%%%%%%%%
\begin{lemma} \label{lem:1}
  Alternate path finding algorithm runs in $O(E)$ time.
\end{lemma}

\begin{proof}
This is easy to follow because we are using a breadth first search which runs in linear to the number of edge traversed. In our alternate path determination phase, we are using a breadth first search from $s_{pre}$ or $s_{suc}$ . In the worst case all the edges are traversed by it. Therefore, the time complexity is $O(E)$.
\end{proof}

%%%%%%%%%%%%%%%%%%%%%%%%%%%%%%%%%%

\begin{lemma} \label{lem:2}
Shifting  algorithm for repairing the broken barrier runs in $O(N^{3})$ time.
\end{lemma}

\begin{proof}
Hungarian algorithm \cite{kuhn:1955} is used to determine the minimum cost maximum matching of a weighted bipartite graph which runs in $O(N^{3})$.
\end{proof}

%%%%%%%%%%%%%
\begin{theorem}\label{th:2}
The overall time complexity of our restoring phase algorithm is $O(N^3)$.
\end{theorem}

\begin{proof}
Restoring phase has two sub phases alternate path finding sub-phase and sensor shifting sub-phase. 
In the worst case both of them are executed. Therefore, the total time requirement is $O(E + N^3)=O(N^3)$. 
\end{proof}

%%%%%%%%%%%%%%%%%%%%%%

\subsection{Energy and Multiple Failure Issues}

In practice, it is not possible to ensure barrier reconstruction for any number of node failure using limited battery capacity. Therefore, to extend the life-time of the barrier coverage, barrier nodes are made static (immobile) when its remaining energy is less than some threshold value. After every relocation remaining displacement capacity of a node is reduced by the displacement made by the node. During node shifting phase the active nodes that are already static never participate in the formation of the bipartite graph. Similarly, during the formation of bipartite graph if the weight of an edge between $(s_i, p_j)$ is greater than displacement capacity (proportional to its remaining energy) of $s_i$ then we remove the edge. If multiple nodes fail simultaneously then the alternate path is determined between the predecessor of the left most failed node and the successor of right most failed node among the barrier nodes. If alternate path is not found between predecessor and successor of the failed node/s then our bipartite matching based node shifting algorithm tries to refill multiple vacant positions created by failed barrier nodes by shifting neighboring barrier as well as non-barrier nodes.

%%%%%%%%%%%%%%%%%%%%%%%%%%%%%%%%%%%%%%%%%%%%%%%%%%%
\section{Distributed Algorithm (DBCRA) }
\label{sec:dist_implementation}

% Restricted Model : The existence of alternate path is determined within a limited hop distance between predecessor and successor of the failed sensor on the barrier. As well as sensors are not allowed to move any arbitrary position but they are relocating themselves to one of their neighbor's position within their relocation capacity.

In this section, we discuss a localized algorithm and its pseudo-code. We assume that each sensor node knows its position, neighbors positions and the positions of left and right boundary. Each node maintains the list of variables in table \ref{tab:variable}. We also assume that sensor fails one after another. It means second fails after the recovery of the first failure. In the rest of this paper, we refer a sensor node by $p$ and a variable $x$ of node $p$ as $p.x$. Following two subsections describe the localized implementation of initial phase and maintenance phase respectively.

%%%%%%%%%%%%%%%%  List of variables  %%%%%%%%%%%%%%%%%%
\begin{table}[ht]

\processtable{Variables at each node $p$ \label{tab:variable}}
{
\begin{tabular}{|l|l|} 
\hline
  {\em id} & The identifier of the node\\ 
\hline
 {\em isOnBarrier} & Boolean flag indicating whether the node is on the barrier \\ 
\hline
 {\em isRecNode} & Boolean flag indicating whether the node is a recovery node \\ 
\hline
  {\em barNeighbor} & List of barrier neighbors  \\ 
\hline
  {\em nonBarNeighbor} & List of non barrier neighbors \\ 
\hline
  {\em pathLength} & Distance along barrier to the closest {\em non-barrier} node \\ 
\hline
{\em resEnergy} & Residual energy of the node \\ 
\hline
  {\em pre} & The predecessor node on the barrier path \\ 
\hline
 {\em suc} & The successor node on the barrier path \\ 
\hline
 {\em recNode} & If {\em isOnBarrier} then it stores the identifier \\
		    & of its recovery node. \\ 
\hline
{\em recSet}  & If {\em isRecNode} then it stores a list of tuples of the  \\
                      & form $\tuple{q, pre, suc}$ where $q$ is a barrier node for which \\
                      & the node is a recovery node, and $pre$ and $suc$ are the  \\
                      & predecessor and successor of $q$. \\ 
\hline
\end{tabular}}{}
\end{table}

% {\em maxHopCount} & Denotes the maximum number of sensors allowed to move for a node failure \\ \hline

\subsection{Initialization Phase}

In the initialization phase, distributed breadth first search is used to find a path between $p_L$ and $p_R$ in the intersection graph. All nodes on the path are marked as {\em barrier nodes}. In Figures \ref{fig:barrier}(a) and (b), the nodes on the  path marked by the dark lines are the {\em barrier nodes}. The remaining nodes are treated as  {\em non-barrier nodes}. Each barrier node $s_i$ keeps track of its \emph{predecessor} and \emph{successor} on the barrier, which are barrier nodes to the immediate left and right of $s_i$ respectively. It also maintains a {\em pathLength}, which keeps the path length along the barrier to its closest non-barrier neighbor. Each barrier node  $s_i$ selects one of its neighbors as a {\em recovery node} (this may be a barrier or a non-barrier node) and sends to it the identifiers and the positions of its predecessor and successor nodes. The recovery node corresponding to node $s_i$ is responsible for detecting the failure of $s_i$ and initiating the repair of the barrier. If $s_i$ has non-barrier neighbors with sufficient energy level to relocate itself to $p_i$, then the closest among them is selected as the recovery node $s_{rec}$ and the {\em pathLength} value of $s_i$ is set to $|p_i - p_{rec}|$. For example, in Figure \ref{fig:barrier}(a), $s_f$ selects $s_{rec}$ as a {\em recovery node}. Otherwise, either the predecessor or successor of $s_i$ is selected as recovery node based on the $pathLength$ and remaining energy of them. If ($|p_i-p_{pre}| + s_{pre}.pathLength$)  $\le$  ($|p_i-p_{suc}| + s_{suc}.pathLength$) then $s_{pre}$ is selected as recovery node of $s_i$; otherwise $s_{suc}$. The {\em pathLength} of $s_i$ is set to $|p_i-p_{rec}| + s_{rec}.pathLength$. 

In Figure \ref{fig:barrier}(b), $s_f$ does not have any non-barrier neighbor but the closest non-barrier node of $s_{pre}$ is $s_6$ whereas closest non-barrier node of $s_{suc}$ is $s_{10}$. Therefore, $s_f$ will select $s_{suc}$ as its recovery node, since  $(|p_f-p_{suc}| + |p_{suc}-p_{10}|) \le (|p_f-p_{pre}| + |p_{pre}-p_5| + | p_5-p_6|)$. The recovery node $s_{rec}$ maintains the position $p_i$ and identifier of $s_i$, as well as the positions and identifiers of the predecessor and the successor of $s_i$. Note that a recovery node $s_{rec}$ can be the recovery node for more than one barrier node, and will maintain the above data for each of them.

\begin{figure}
\centering
\includegraphics[width=11cm]{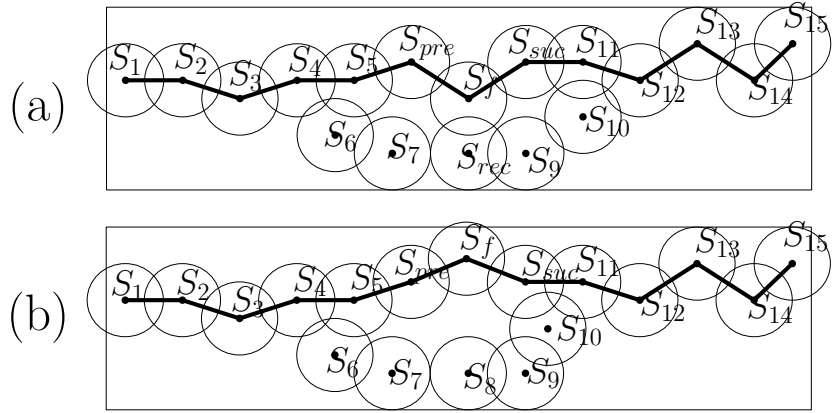}
\caption{Example of barrier nodes and recovery node}
\label{fig:barrier}
\end{figure}

As discussed earlier, the initialization phase finds the initial barrier using a distributed BFS, marks the barrier nodes, and subsequently sets up the recovery node for each of the barrier nodes. Since distributed BFS algorithms are well-known, we assume that the initial barrier is set up already and the $isOnBarrier$ variable is properly initialized at each node. In addition, if $isOnBarrier = true$ for a node, then we assume that its $pre$ and $suc$ variables are also properly initialized. We just show the pseudo code for setting up the recovery node for each barrier node.

Three types of messages are used to set up the recovery nodes:

\begin{itemize}

     \item {\em ReqNbRec(q)}: A request message sent by a barrier-node  to its predecessor and successor to find the path length of the closest non barrier node. The parameter $q$ is the identifier of the message originator.

     \item {\em RepNbRec(q,d)}: A  reply message returned in response to a {\em ReqNbRec(q)} message. The two parameters $q$ and $d$ indicate that the destination for the reply message is $q$ and the path length to the discovered closest non-barrier node is $d$.

     \item {\em SetRec(p,q)} : A message sent by a barrier-node to its recovery node. the parameter $p$ contains the identity and the position of its predecessor and the parameter $q$ contains the identity and the position of its successor.

\end{itemize}

The algorithms for the initialization phase are shown in Algorithm \ref{algo:initial1} and Algorithm \ref{algo:initial2}. We assume that the initial barrier construction algorithm identifies the barrier nodes and sets predecessor and successor nodes of each individual barrier node on the current barrier. We also assume the initial energy level of all the sensors are same and spent according to the displacement they make. Algorithm \ref{algo:initial1} shows the pseudo-code executed by each barrier node to find its recovery node. If there is a non-barrier neighbor then the barrier node sends a {\em SetRec} message to it, otherwise chooses one of its barrier neighbors by sending {\em ReqNbRec} messages to them.  Algorithm \ref{algo:initial2} shows the messages handling subroutines of the initialization phase.

%%%%%%%%%%%%%%%%%%%%%%%%%%%%%%%%%%%%%%%%%%%%%%%%%%

 \begin{algorithm}%[!h]
% \SetLine

 {\bf Initial Values} \\

 $pathLength = \infty$ ;  
 %$resEnergy=Emax$ \;
 $barNeighbor, nonBarNeighbor, recSet = \emptyset$

\vspace{.3cm}

\SetKwFunction{KwFn}{ {\bf Barrier Node $p$ : After completion of initial barrier construction\vspace{8pt}  } }
\KwFn

	  \tcc{Initialize phase for barrier nodes to set their recovery node}

     	  $p.nonBarNeighbor = $ non-barrier neighbors of $p$ with enough energy for relocating to $p$ \;
	  $p.barNeighbor = $ barrier neighbors of $p$\;
	  $l = p.pre$\;
	  $r = p.suc$\;

	  \uIf{ $|p.nonBarNeighbor| > 0 $}
	    {
	    $t$= Closest node among $p.nonBarNeighbor$ \;
	    
	    $p.pathLength = distance(p,t)$;  \tcc{ Eucledian distance between $p$ and $t$ } \ 
	    $p.recNode = t$ \;
	    
	    send {\em SetRec(l,r) } to $t$;  \tcc{ Inform the recNode the positions and identity of $l$ \& $r$ }
	    }
	  \Else
	    {
	      \tcc{Send request to $l$ \& $r$ to know distance of their recNodes}
	      send {\em ReqNbRec(p)} to $l$  \;
	      send {\em ReqNbRec(p)} to $r$  \;
	    }

\caption{Localized Barrier Maintenance Algorithm: Initialization Phase}
\label{algo:initial1}
\end{algorithm}

\begin{algorithm}%[!h]

 \vspace*{.6cm}

 \SetKwFunction{KwFn}{ {\bf Upon receiving a {\em SetRec( pred, succ )} message from node $r$} \vspace{8pt} }
 \KwFn
    
  $p.recSet = p.recSet \cup  \{  \tuple{r, pred, succ}  \} $   \;

%%%%%%%%%%%%%%%%%%%%%%%%%%%%%%%%%%%%%%%%%%%%%%%%%
\vspace*{.6cm}
 \SetKwFunction{KwFn}{ {\bf Upon Receiving a {\em ReqNbRec(q)}  message from node $p.suc$} \vspace{8pt} }
 \KwFn

    \If {$p.resEnergy < $ required energy to move to the location of $p.suc$}
    {

	send {\em RepNbRec(q,$\infty$)} to $p.suc $\;
	\KwRet
    }
    
  $p.nonBarNeighbor$ =  non-barrier neighbors of $p$ with enough energy for relocating to $p$ \;
  
  \uIf { $|p.nonBarNeighbor| > 0 $}
    {
	  $t$= Closest node among $p.nonBarNeighbor $\;
	  $d=distance(p,t) + distance(p,p.suc) $\;
	  send {\em RepNbRec(q,d)} to $p.suc $\;
    }
  \uElseIf{ $|p.nonBarNeighbor| = 0  \wedge  p.recNode = p.suc$ }
    {
	  send {\em ReqNbRec(q)} to $p.pre$ \;
    }
  \Else
    {
	  $d=p.pathLength + distance(p,p.suc) $\; 
          send {\em RepNbRec(q,d)} to $p.suc$\;
    }

%%%%%%%%%%%%%%%%%%%%%%%%%%%%%%%%%%%%%%%%%%%%%%%%
\vspace*{.6cm}

\SetKwFunction{KwFn} { {\bf Upon Receiving a {\em RepNbRec(q,d) }  message from $p.pre$} \vspace{8pt} }
\KwFn

    \uIf{ $q=p$ }
        {
	\If{ $p.pathLength > d$ }
	      {
		  $p.pathLength = d$\;
		  $p.recNode =  p.pre$\;
                send {\em SetRec( p.pre, p.suc  ) } to $p.pre$\;
              }
        }
     \Else
       {
	   $p.pathLength = d$\;
	   $d=d + distance(p,p.suc) $\; 
	   {\em p.recNode} =  {\em p.pre} \;
	   send {\em RepNbRec(q,d)} to $p.suc$\;
       }

\tcc{ Similar types of actions are carried out when $p$ receives {\em ReqNbRec(q)} message from $p.pre$ and {\em RepNbRec(q,d)} message from  $p.suc$ ($suc$ is replaced by $pre$ and vice versa).}

 \caption{ Handling of different types of messages by a node $p$ in the Initialization Phase}
 \label{algo:initial2}
 \end{algorithm}

\subsection{Maintenance Phase}

In the maintenance phase, on failure of a barrier node, it is first checked if an alternate barrier exists or not. The recovery node $s_{rec}$ of the failed node $s_f$ detects the failure of $s_f$, and attempts to discover an alternate path between the predecessor and the successor of $s_f$. The alternate path is determined using  a localize heuristic search method which is discussed  in the following subsection \ref{subsec:alternate_path_localize}.  If no such path is found within a given maximum hop count, then the algorithm attempts to replace the failed node by moving few neighboring sensors to new locations. 

Note that no maintenance needs to be done if a non-barrier and non-recovery node fails, as the barrier is not affected by the failure. If a non-barrier recovery node $s_{rec}$ fails, each barrier node for which $s_{rec}$ is the recovery node detects the failure, and restarts the process of finding a new recovery node using the scheme described in the description of the initialization phase. The $pathLength$ value of each such node may also change, which is then propagated through the predecessor and the successor of the node, as that value may have been used more than one hop away to set the $pathLength$ value of nodes. This may in turn change the recovery node for other barrier nodes.

After the failure of a {\em barrier node} $m$ and $m \in$ {\em recSet(p)}, $p$ is responsible to repair the barrier. The detail algorithm for handling of node failure in maintenance phase is shown in Algorithm \ref{algo:maintenance}. In the maintenance phase {\em Alternate-Path($d$)} procedure is called to determine a path from a given node to $d$. To determine the path it uses the following message.

\begin{itemize}
\item {\em $Tok(route,d,k)$} message to find a path for node $d$ within $k$ hops. This message is used for both route discovery and route reply. In route discovery mode the parameter $route$ is initialized with $``disc"$, whereas in route reply mode the parameter is initialized with $``found"$.
\end{itemize}

%%%%%%%%%%%%%%%%%%%%%%%%%%%%%%%%%%%%%%%%%%%%%%%%%%%

\begin{algorithm}[!h]
%\SetLine

\tcc{ p is a recovery node of m}
\uIf{ m $\in$  p.recSet }
{

\tcc{ p is a non barrier recovery node}
\uIf { p $\neq$ m.pre $\wedge$  p $\neq$ m.suc } 
     { alternate-barrier = concatenate( {\em Alternate-Path(m.pre), Alternate-Path(m.suc) } ) \; }

\tcc{ p is a predecessor of m}
\uElseIf{  p= m.pre }
     { alternate-barrier= {\em Alternate-Path(m.suc) } \;}

\tcc{ p is a successor of m}
\Else{ \tcc{ {\em p= m.suc}  }
      alternate-barrier = {\em Alternate-Path(m.pre) } \;}

\tcc{ when alternate path is not found within a given maximum depth update the position and energy}
\If {  $\lnot$ {\em alternate-barrier} }
    {
	\uIf{ p.resEnergy $\ge$  energy require to relocate to m}
	{
	  move $p$ to replace $m$ \;
	  update $resEnergy$ of $p$ \;
	}
	\Else{
	  recovery is not possible \;
	  \KwRet \;	    
	 }
    }
}

\tcc{On failure of a recovery node determine new recovery node}
\uElseIf{p.recNode= m }
{
  starts the initialization phase of node $p$ \;
}

\tcc{On failure of other node}
\Else
{
  do nothing \;
}

\caption{Upon Detecting Failure of Node $m$ by Node $p$  }

\label{algo:maintenance}
\end{algorithm}

%\FloatBarrier
%%%%%%%%%%%%%%%%%%%%%%%%%%%%%%%%%%%%%%%%%%%%%%%%%%%

\subsubsection{Determination of Alternate Path }
\label{subsec:alternate_path_localize}
In this subsection, we discuss our proposed backtracking search based algorithm to find an alternate path between the predecessor $s_{pre}$ and the successor $s_{suc}$ of a failed node $s_f$ in the intersection graph. The propose algorithm, called  {\em Modified Limited Depth First Search} (\emph{MLDFS}), is a combination of limited depth first search (\emph{LDFS}) (a depth first search limited to a given maximum depth) and Finn's greedy routing (\emph{FGR}) algorithm \cite{finn:1987}. The limited depth first search is used to limit the hop count of the path found, while the greedy routing strategy helps to keep the new path geographically close to the part of the path that is lost due to the failure of a barrier node. After the failure of a barrier node $s_f$, \emph{MLDFS} is used to find an alternate path between the $s_{pre}$ and $s_{suc}$. \emph{MLDFS} first selects a neighbor based on \emph{FGR}. In particular, it selects an unexplored neighbor $s_x$ which is closest to the destination $s_d$. An example of this neighbor selection using the \emph{FGR} algorithm is shown in Figure \ref{fig:neighbor}(a) where node $s_{rec}$ wants to find a path to node $s_d$. $s_{rec}$ selects the neighbor $s_a$ first instead of $s_b$, because $s_a$ is closer to $s_d$ than $s_b$. If there is no such path found through $s_a$ then it backtracks and tries through $s_b$. A complete alternate barrier of Figure \ref{fig:graphical_representation}(a) after the failure of a barrier node $s_f=s_8$ is shown in Figure  \ref{fig:neighbor}(b) where $s_{pre}= s_6$ and $s_{suc}=s_9$.

\begin{figure}
\centering
\includegraphics[width=12cm]{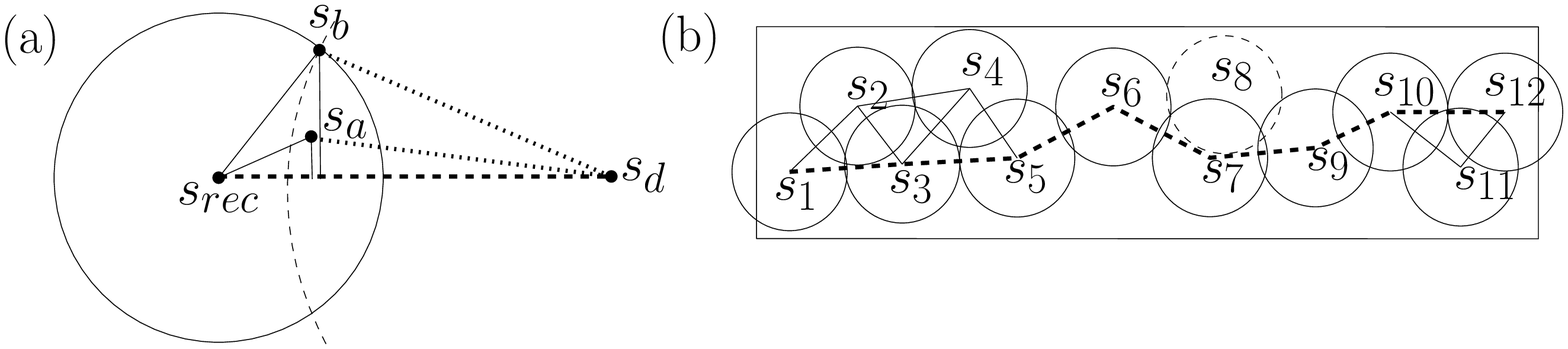}
\caption{Neighbor selection and alternate path determination}
\label{fig:neighbor}
\end{figure}
%\vspace{0.15cm}

If the recovery node $s_{rec}$ is a non-barrier node, then it finds two individual paths to predecessor $s_{pre}$ and the successor $s_{suc}$ of the failed node $s_f$, and then combines them with the rest of the broken barrier (from $p_L$ to $s_{pre}$ and from $s_{suc}$ to $p_R$) to form an alternate barrier. For example in Figure \ref{fig:recovery}(a), on failure of node $s_f$, $s_{rec}$ (a non-barrier recovery node) finds two paths from $s_{rec}$ to $s_{pre}$ and $s_{rec}$ to $s_{suc}$ using Algorithm \ref{algo:alternatePath} (\emph{MLDFS}), and merges them to the earlier barrier to construct the alternate barrier $\{s_1, s_2, s_3, s_4, s_5, s_6, s_7, s_{rec}, s_9, s_{10}, s_{11}, s_{12}, s_{13}, s_{14}, s_{15} \}$. On the other hand, if $s_{rec}$ is itself a barrier node, it uses algorithm \emph{MLDFS} to find a path to the other barrier neighbor of $s_f$. In Figure \ref{fig:recovery}(b), $s_{suc}$ is a barrier recovery node of the failed sensor $s_f$. $s_{suc}$ uses \emph{MLDFS} to find a path to the other barrier neighbor $s_{pre}$ of $s_f$, and combines it with the previous barrier to form an alternate barrier $\{s_1, s_2, s_3, s_4, s_5, s_6, s_7, s_8, s_9, s_{10}, s_{suc}, s_{11}, s_{12}, s_{13}, s_{14}, s_{15} \}$.

The detail algorithm for determining alternate path is shown in Algorithm \ref{algo:alternatePath}. The recovery node of the failed barrier is the initiator. It initiates the path discovery process by setting route variable to $``disc"$ and sending {\em Tok(route,d, k)} message to one of its neighbor. The neighbor selection process is based on the distance to the destination node $d$'s position \cite{finn:1987}. The route discovery process is divided into four cases depending on the receiver of the {\em Tok}. First, the receiver is the final destination.  Second, the receiver is an intermediate node which has at least one neighbor except from where it receives the message but the token is not yet forwarded to it. Third,  once the token is already reached its limited depth then the receiving node returns the message to the sender. Finally, there does not exist a neighbor of the receiver to which the token message forwarding is still pending then it declares failure of route discovery. Once the destination node $d$ receives the route discovery message it sets the $route$ variable to $``found"$ and returns back to the node from where it receives. Similarly, once an intermediate node receives a route $``found"$ message, it forwards the message to the node from where it first received (father). This process continues until it reaches to its originator.  

\begin{figure}[!h]
\centering
\includegraphics[width=12cm]{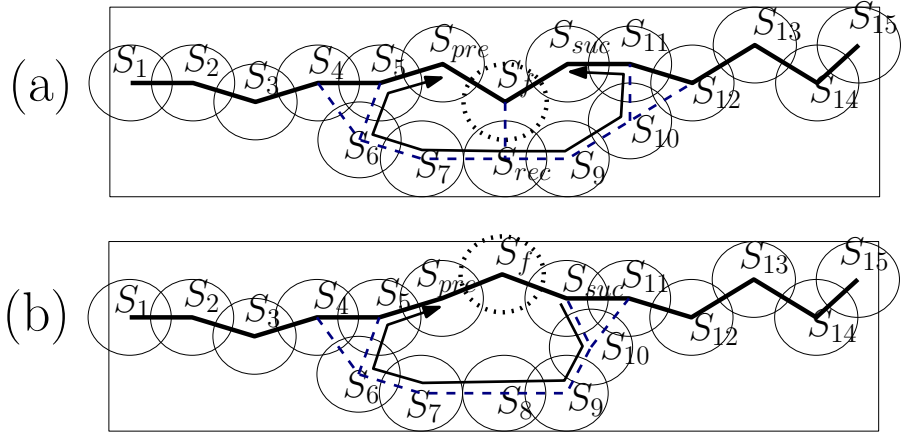}
\caption{Discovery of alternate path} \label{fig:recovery}
\end{figure}

\begin{algorithm}[!ht]
%\SetAlgoLined
%\SetLine

{\bf var} \\

\begin{tabular}{l c l l l}
    {\em neighbor }  &:& node set      & $\quad$ init  & {\em barNeighbor} $\cup$  {\em non-barNeighbor} \\
    {\em used(q) }   &:& boolean array & $\quad$ init  & false for each $q$ $\in$ {\em neighbor(p)} \\
    {\em father }    &:& node          & $\quad$ init  & udef
\end{tabular}

\vspace{0.5cm}

\SetKwFunction{KwFn}{{\bf Initiator Node $p$ :} \vspace{8pt} }
\KwFn

    {\em p.father = p} \;
    
    choose $q$ $\in$ {\em p.neighbor}  \;   
    \tcc{In this algorithm $q$ is chosen such that $\overline{qd}$ is minimum among all neighbors of $p$ \cite{finn:1987}}
    {\em p.used(q)}  = true \;
    send {\em Tok(``disc",d, k)} to $q$ \;

\vspace{.5cm}

\SetKwFunction{KwFn}{{\bf Any Node $p$ : Upon Receiving a {\em Tok(route,d,k) } message from $m$ } \vspace{8pt}}
\KwFn

\If{ p.father= udef}{ {\em p.father} = $m$ \; }

\uIf{ route = ``disc"   $\wedge$  (d = p)  $\wedge$ (k $\ge$ 0) }
	{    \uIf{ p.father= p.pre $\vee$ p.father = p.suc}{ isOnBarrier=false; \hspace{.5cm} p.pre=udef; \hspace{.5cm} p.suc=udef;}
	     \Else
		{ 
		  update {\em p.pre} or {\em p.suc} according to destination $d$ and $m$;
		}
	     send {\em Tok(``found", null,1)} to $m$ \;
	}

\uElseIf { route = ``disc" $\wedge$  k $>$ 0  $\wedge$  $\exists q \in$ p.neighbor $\backslash$ p.father : $\lnot$ p.used(q) }
{

  \uIf { p.father $\ne$ m  $\wedge$  $\lnot$ p.used(m)  }
   {
      $q=m$ \;
   }
  \Else
   {
      choose $q \in$ {\em neighbor(p)} $\backslash$ {\em p.father} : $\lnot$ {\em p.used(q)} $\wedge  \overline{qd} $ is minimum \;
   }

   {\em p.used(q)} = true \;
   send  {\em Tok(``disc",d,k-1) } to $q$ \;
}

\uElseIf { route = ``disc" $\wedge$  k $=$ 0 } 
{

  {\em p.used(p.father) }= true; \;
  send {\em Tok(``disc",d,k+1)} to $p.father$ \;

}

\uElseIf { route = ``disc"  $\wedge$  $\forall q \in$ p.neighbor : p.used(q)=true }
{ 
  path discovery failure \;  
}

\uElseIf { route = ``found" $\wedge$ p.father $\neq$ p }
	{
	     \uIf{ p.father= p.pre $\vee$ p.father = p.suc }{ isOnBarrier=false; \hspace{.5cm} p.pre = udef; \hspace{.5cm} p.suc=udef;}
	     \Else
	      { 
		update  {\em p.pre} or {\em p.suc} according to destination $d$ and $m$; 
	      }

	      send {\em Tok(``found",d,k+1)}  to  {\em p.father } \;
	}
	
\uElseIf { route = ``found"  $\wedge$  p.father=p }
 	 {	
	  path discovery success \; 
	}

\caption{ Discover {\em Alternate-Path } to Node $d$ within $k$ Hops ({\em MLDFS}) }
\label{algo:alternatePath}
\end{algorithm}

\subsubsection{Shifting Sensors to Reestablish Barrier}
\label{subsec:sensor_move_localize}
If no alternate path is found between $s_{pre}$ and $s_{suc}$ in the intersection graph, the barrier cannot be repaired by our alternate path phase. This can be due to two reasons, no alternate path actually exists, or there is a path but its length is too large. In both cases, the barrier coverage maintenance algorithm assumes that alternate path can't be found, and hence tries to shift sensors to new positions to reestablish the barrier. 
\FloatBarrier

We try to locally repair the barrier by trying to exactly recreate the barrier that was present before the failure. In the initialization phase recovery nodes of the barrier nodes are chosen such that through them the distance to non-barrier node is minimum. Therefore, the recovery node provides a shifting direction for every failed barrier node with minimum total displacement. In particular, the recovery node $s_{rec}$ of the failed node $s_f$ moves to the position of $s_f$ to breach the hole created by the failure of $s_f$. If $s_{rec}$ is a non-barrier node, the original barrier is recreated and the algorithm terminates. If $s_{rec}$ is itself a barrier node, moving $s_{rec}$ from its original position may now break the barrier again. Note that the state of the barrier after moving $s_{rec}$ is indistinguishable from the state in which the node $s_{rec}$ failed. Thus, if the barrier is broken due to movement of $s_{rec}$, it can be considered as the failure of $s_{rec}$. This is now detected by another recovery node $s_h$, which is a recovery node of $s_{rec}$, and triggers the same action by $s_h$. The process continues until the recovery node is a non-barrier node. The pseudo-code of the shifting part is included in Algorithm \ref{algo:maintenance}.

%We restrict the maximum number of nodes that can be moved to some constant $k$. If a non-barrier node is found
%within $k$ shifts, then the barrier is reestablished, otherwise the algorithm fails to recover the barrier.

%%%%%%%%%%%%%%%%%%%% Complexity Analysis %%%%%%%%%%%%%%%%%%%%%%%%%

\subsection{Complexity Analysis}
\label{sec:analysis}
In this section, we establish the time and message complexity of different phase of the above distributed protocol.

\begin{theorem} \label{th:1}
The distributed initialization phase of barrier maintenance algorithm runs in $O(N)$ round time with $O(NE)$  message overhead.
\end{theorem}

\begin{proof}
The initialization phase uses breadth first search to find a path between $p_L$ and $p_R$ and establish the barrier. After initial barrier construction, barrier nodes determine their closest non-barrier recovery node. Together, this phase runs in $O(N)$ round and with $O(NE)$  message. 
\end{proof}

%%%%%%%%%%%%%%%%%%%%%%%%%%%%%%%%%%
\begin{lemma} \label{lem:1}
 Alternate path finding algorithm runs in $O(k)$ time and message overhead.
\end{lemma}

\begin{proof}
This is easy to follow because an intermediate node including the recovery send $Tok()$ message to one of its neighbor based on the destination node position.  In worst case the $Tok()$ will traverse at most $k$ hop distance.  Therefore, time and message complexity is $O(k)$.
\end{proof}

\begin{lemma} \label{lem:2}
Shifting  algorithm for repairing the broken barrier runs in $O(N)$ round time with $O(NE)$  message overhead, where $E$ is the number of edges in the communication graph of the sensor network.
\end{lemma}

\begin{proof}
After failure of a barrier node, the corresponding recovery node will shift to the failed sensor position in cascaded fashion. Therefore, in worst case all the nodes on the barrier will move to one of its neighbor node position. Hence, at most $O(N)$ number of movement will be there and it will take $O(N)$ rounds of operation. Assuming in a single round a sensor can move to any one of its neighboring sensor position. After shifting, the barrier nodes will again rediscover their recovery nodes, which will take $O(NE)$  message overhead. 
\end{proof}

\begin{theorem}\label{th:2}
The overall time  and message complexity of our distributed  restoring phase algorithm is $O(N)$ and $O(NE)$ respectively.
\end{theorem}

\begin{proof}
Restoring phase has two sub phases alternate path finding sub-phase and sensor shifting sub-phase. Together, 
it runs in $O(k +N) = O(N)$ times and $O(k + NE)= O(NE)$ message overhead.
\end{proof}

%%%%%%%%%%%%%%%%% END Complexity Analysis %%%%%%%%%%%%%%%%%%%%

%\FloatBarrier
\section{Experimental Results}
\label{sec:results}

We perform detailed simulations to evaluate the performance of the proposed algorithm. In our simulation, we have used $N$ = 140, 160, 180 number of sensors. The sensors are dropped with uniform separation along a horizontal line within a rectangular barrier of length $L=4000$ units. The distance between two consecutive sensors is supposed to be $\psi = \frac{L}{(N-1)}$ from the left boundary to the right boundary. However, due to air flow, deployment error, and other reasons the sensors may not fall exactly where they are targeted. In the simulation, a random offset with a standard deviation $\delta=6$ units from the target position is considered as a deployment error as in \cite{saipulla:2009}. We assume uniform sensing range of $\rho=30$ units and the communication range to be twice of the sensing range. It is found experimentally that with this sensing range and the above deployment model, choosing $N < 140$ makes the intersection graph even disconnected in many cases and fails to form a barrier initially. Also, choosing $N = 180$ makes  the intersection graph quite dense which forms many alternate barriers and many non-barrier nodes. Hence, in the experiment we restrich the number of sensors  within 120 to 180.

In our experiments, the maximum depth allowed to find alternate path is limited to $k=N/20$. Sensors are started with an initial energy of $100$ units. The migration energy has been assumed to be $1$ unit per unit displacement. We start with an initial deployment that ensures the formation of an initial barrier and sensors are failed one by one at a time uniformly. Sensors are chosen randomly for failure and the barrier is reconstructed (if necessary and possible) before the failure of the next sensor.

The following metrics are measured to evaluate the performance.

\begin{itemize}

\item {\em Recovery rate}: This measures the ability of an algorithm to reestablish the barrier.   If $x$ sensors are failed one by one, and a barrier is reestablished for $y$ times $(y \leq x)$ of those $x$ failures,   then the recovery rate is defined as $\frac{y}{x}\times 100$. Thus it is the percentage of failures for which a barrier is   reestablished.

\item {\em Remaining energy}:  We measure the remaining energy level of sensors and show the percentage of the total number of sensors whose energy level is high (more than 90\% of their initial energy level) after final recovery. 

%	 High : remaining energy in in between 71 to 100\% of the initial energy. 
%      Medium: remaining energy in in between 31 to 70\% of the initial energy.
%	 Low : remaining energy in in between 0 to 30\% of the initial energy.

\item {\em Average total displacement}: This is another measure to find the cost incurred by an algorithm. It is the average of the total displacement made by all sensors to reestablish the barrier after each failure.
\end{itemize}

Each of these measures are computed over 100 runs and the results presented in this section are the average of these runs. We have shown the performance metric measures of the algorithm not after every individual node failure but after every 5\% node failure in between [5 to 30]\% of the total number of sensors present initially. In the rest of this section, we refer to our proposed cenralized algorithm as {\em C-Move} and distributed algorithm as {\em D-Move} for barrier coverage maintenance. Due to unavailability of any existing barrier maintenance algorithm in literature, we compare our algorithm's performance  with the following schemes.

\begin{enumerate}

\item  {\em No Move} ({\em N-Move}): In this scheme, sensors are static, and after every failure a centralized algorithm is used to find an alternate barrier. This algorithm acts as a very basic reference for comparison.

\item {\em Random Move} ({\em R-Move}): This is a variation of {\em D-Move} in which no explicit recovery node is maintained. On the failure of a sensor $s_i$, if $s_i$ has any non-barrier neighbor $s_j$ with enough energy to relocate to the position $p_i$, then $s_j$ is moved to replace $s_i$. If no such $s_j$ exists, any one of the predecessor or successor of $s_i$ with enough energy to relocate is selected at random and moved to take the position of $s_i$. The barrier nodes are then shifted in cascaded fashion in the same direction until a non-barrier neighbor is found.

\end{enumerate}

\begin{figure}[!h]
%\vspace{0.4in}
\centering

\includegraphics[angle=-90, width=5cm]{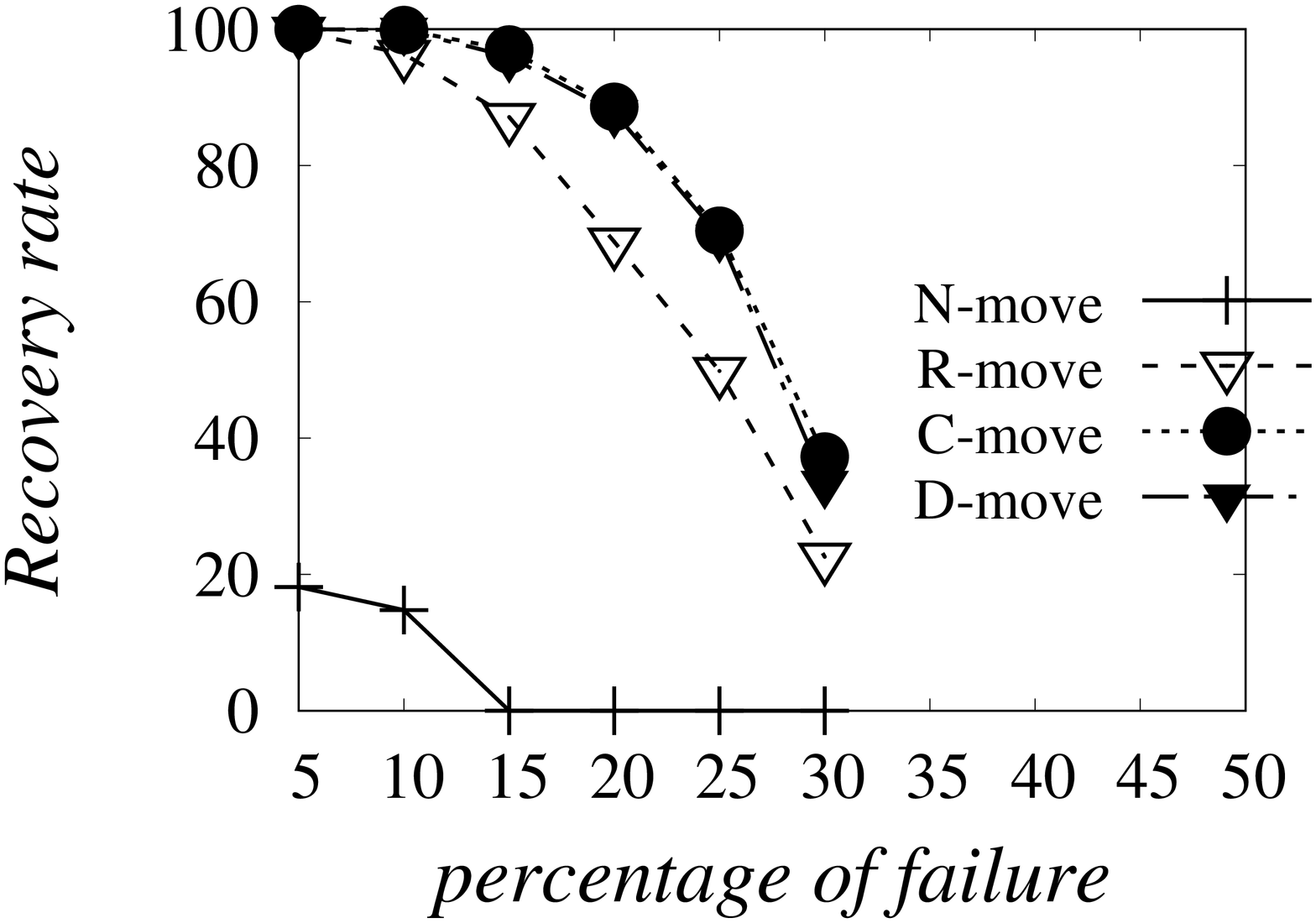}
\includegraphics[angle=-90, width=5cm]{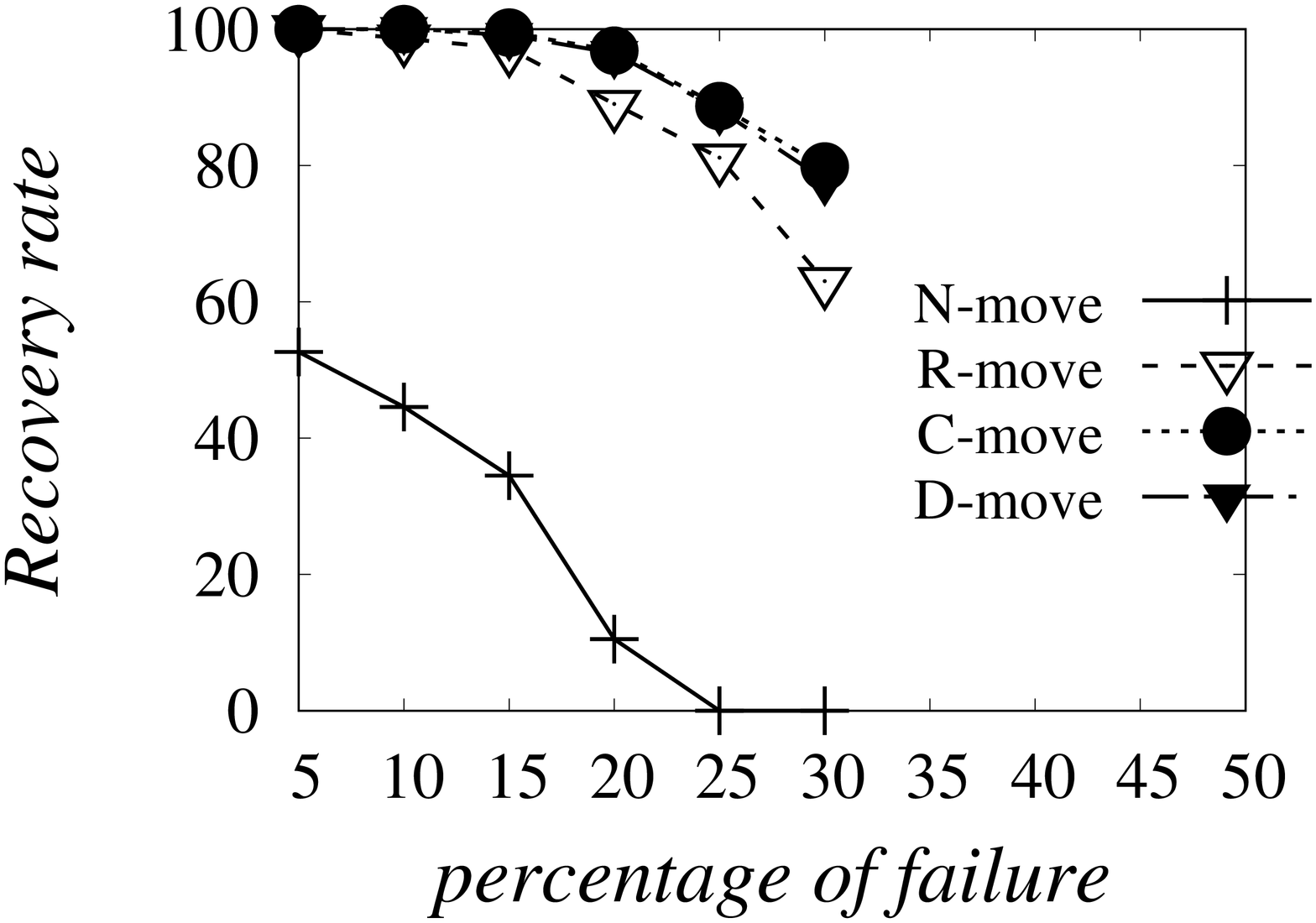}
\includegraphics[angle=-90, width=5cm]{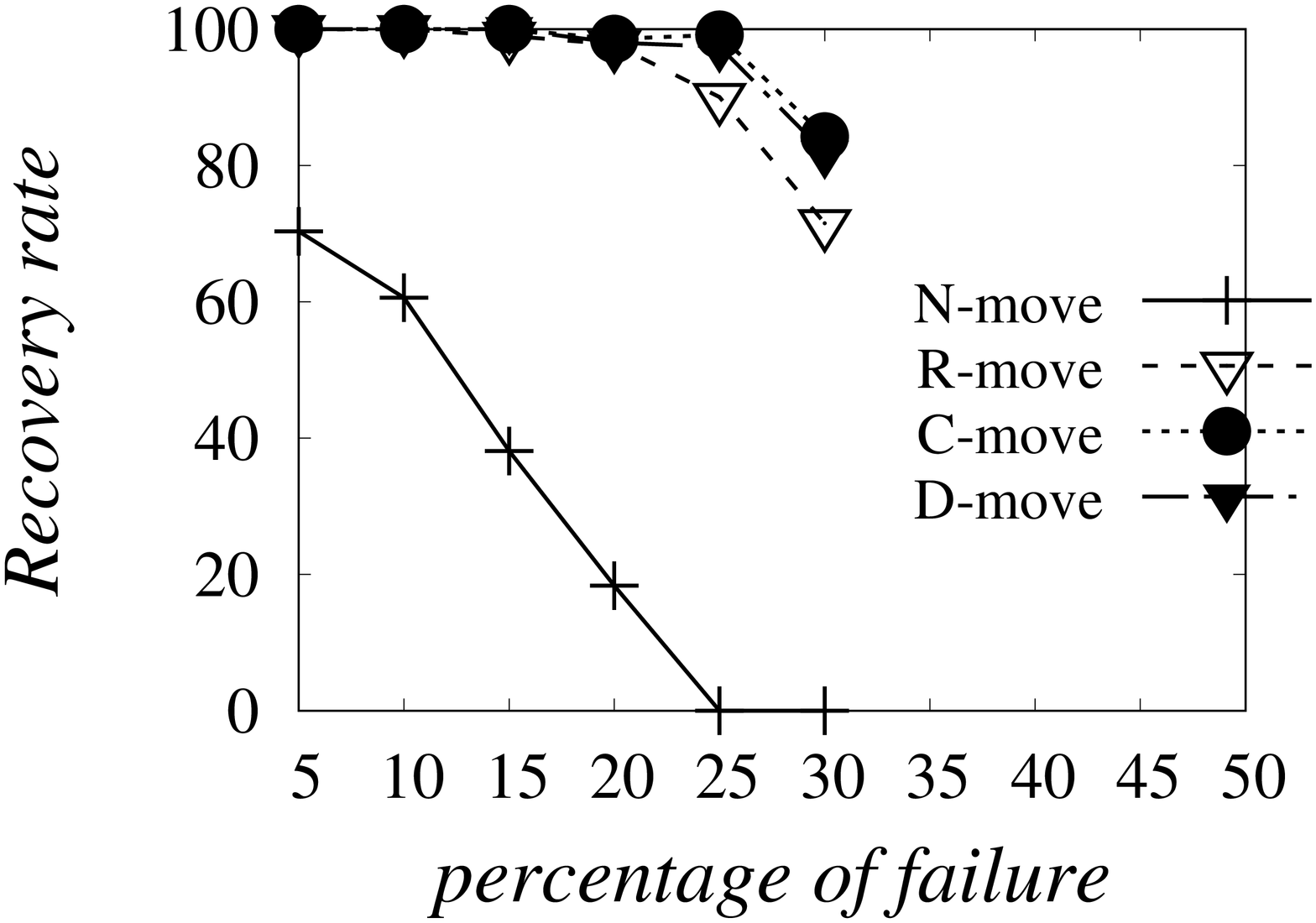}

\caption{Average recovery rate with varying number of sensors deployed (a) N=140   (b) N=160     (c) N=180 }
\label{chart:avg_rec}
%\vspace{0.3in}
\end{figure}

%The results for {\em MEBC} and {\em MRMR} are not shown as they achieve 100\% recovery rate for the range of failures considered.

Figure \ref{chart:avg_rec} shows the variation of the average recovery rate of the algorithms with different percentages of failure. As expected, it is seen from the figure that if no sensors are allowed to be moved ({\em N-Move}), the recovery rate is relatively low even for a small percentage of sensor failures, going down to near 0 even with [15 to 25]\% failures. This is true even with $N=180$, though the recovery rate is slightly higher in this case due to the larger number of edges in the intersection graph.  But, this shows that finding alternate barriers alone is not sufficient to reestablish broken barriers and movement of sensors must be allowed for barrier coverage maintenance. The average recovery rate of {\em R-Move}, {\em C-Move} and {\em D-Move} increases as $N$ increases. This is expected as the number of non-barrier neighbors and the number of alternate paths in the intersection graph increase as $N$ increases. {\em C-Move} and {\em D-Move} show around 5-15\% better performance than {\em R-Move} in terms of recovery rate. The choice of a random direction may cause the cascaded movement of sensors to continue without finding a barrier node with a non-barrier neighbor. However, the use of minimum cost matching based algorithm {\em C-Move} and closest distance non-barrier node based algorithm {\em D-Move} ensure that the direction of cascaded movement chosen has a greater chance of avoiding the above bad case. However, with larger number of sensors, the chance of \emph{R-Move} falling into the above bad case is lowered, and thus \emph{R-Move}, {\em C-Move} and \emph{D-Move} show similar performance.

%\begin{figure}[!h]
%\centering
% \includegraphics[angle=-90, width=3.5cm]{Figure/sigma=9/Rec_rate_for_N120_k=10_30_60.eps}
% \includegraphics[angle=-90, width=3.5cm]{Figure/sigma=9/Rec_rate_for_N140_k=10_30_60.eps}
% \includegraphics[angle=-90, width=3.5cm]{Figure/sigma=9/Rec_rate_for_N160_k=10_30_60.eps}
% 
% \caption{ Average recovery rate with varying $k$ (the number of sensors allowed to be moved) (a)  N=120 (b)  N=140 (c)  N=160}
% \label{chart:avg_rec2}
% \end{figure}

% The performance of \emph{APMB-Move} is also affected by the value of $k$, the number of hops the cascaded movement is allowed
% to go on (equal to the number of sensors allowed to move). Figure \ref{chart:avg_rec2} shows how the recovery rate of
% \emph{APMB-Move} varies with $k$. As expected it is seen that as $k$ increases, the average recovery rate also increases.
% However, beyond a certain percentage of failure, no value of $k$ is useful as there are not enough non-barrier nodes left.

% \begin{figure}[!h]
% \centering
% \includegraphics[angle=-90, width=3.5cm]{Figure/sigma=9/N120sigma9-avg-hop.eps}
% \includegraphics[angle=-90, width=3.5cm]{Figure/sigma=9/N140sigma9-avg-hop.eps}
% \includegraphics[angle=-90, width=3.5cm]{Figure/sigma=9/N160sigma9-avg-hop.eps}
% \caption{ Average number of sensors moved for each failure (a)  N=120  (b)  N=140 (c)  N=160 }
% \label{chart:avg_hop}
% \end{figure}

%% Total displacement discussion

\begin{figure}
\centering

\includegraphics[angle=-90, width=5cm]{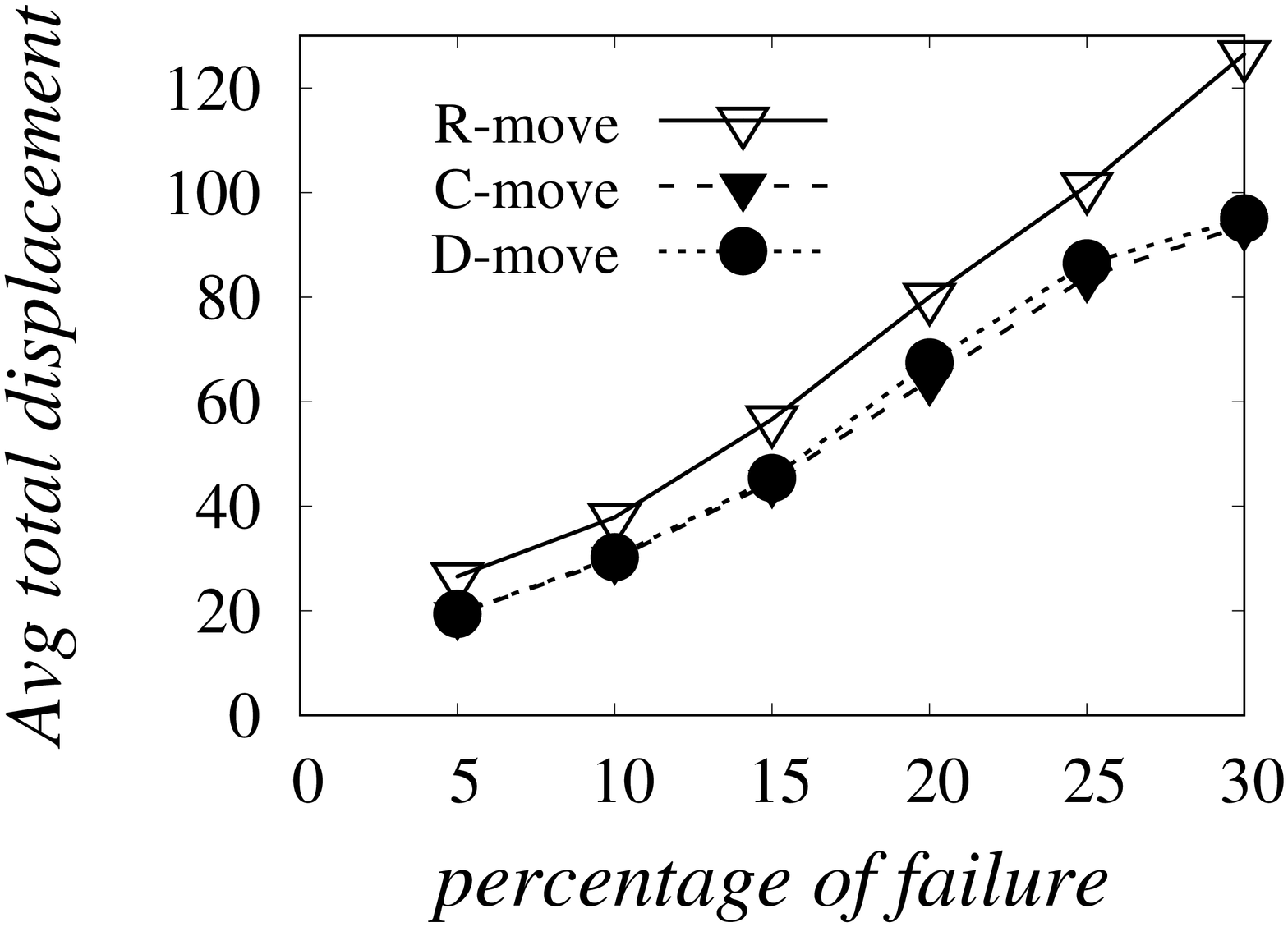}
\includegraphics[angle=-90, width=5cm]{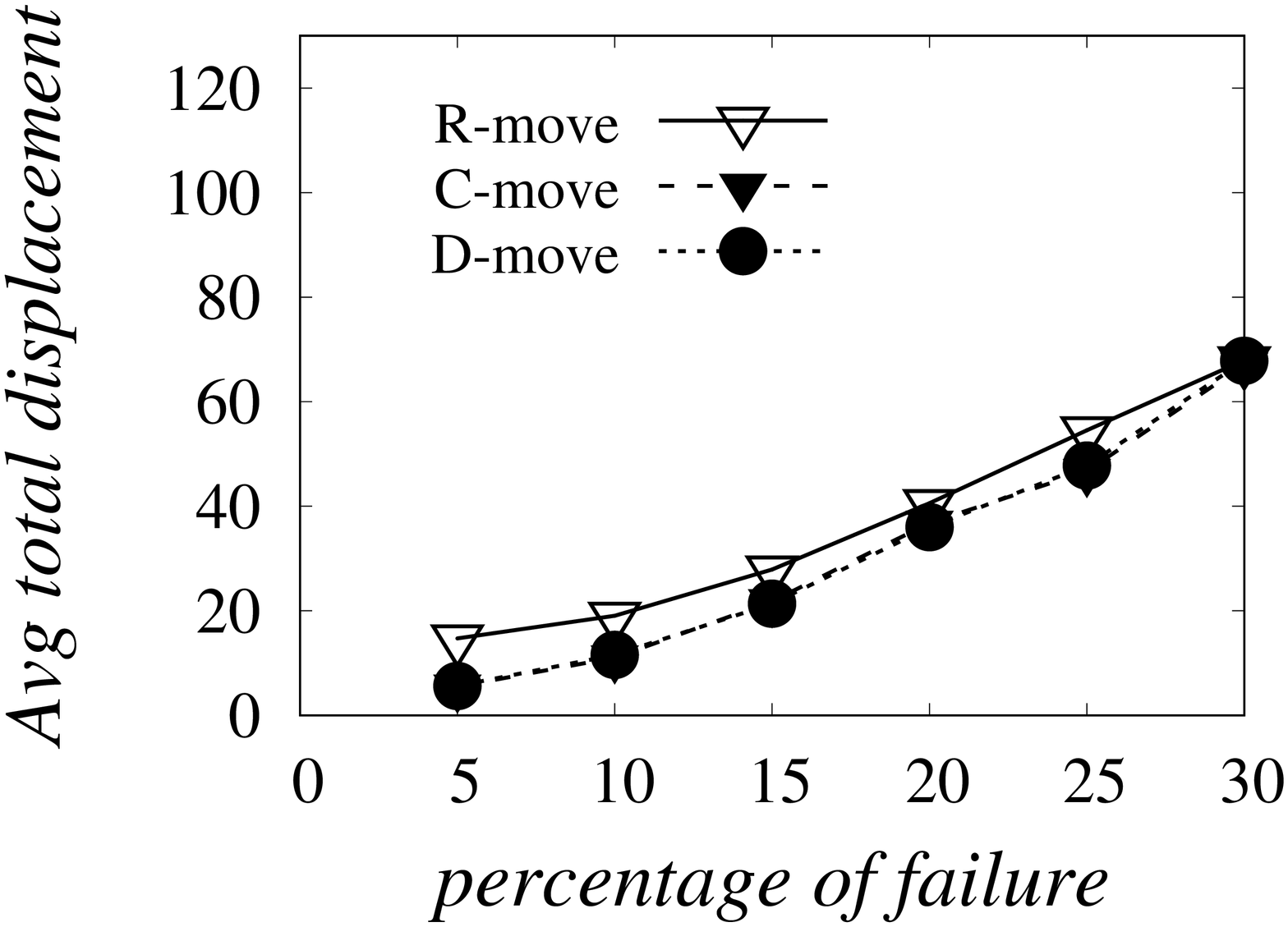}
\includegraphics[angle=-90, width=5cm]{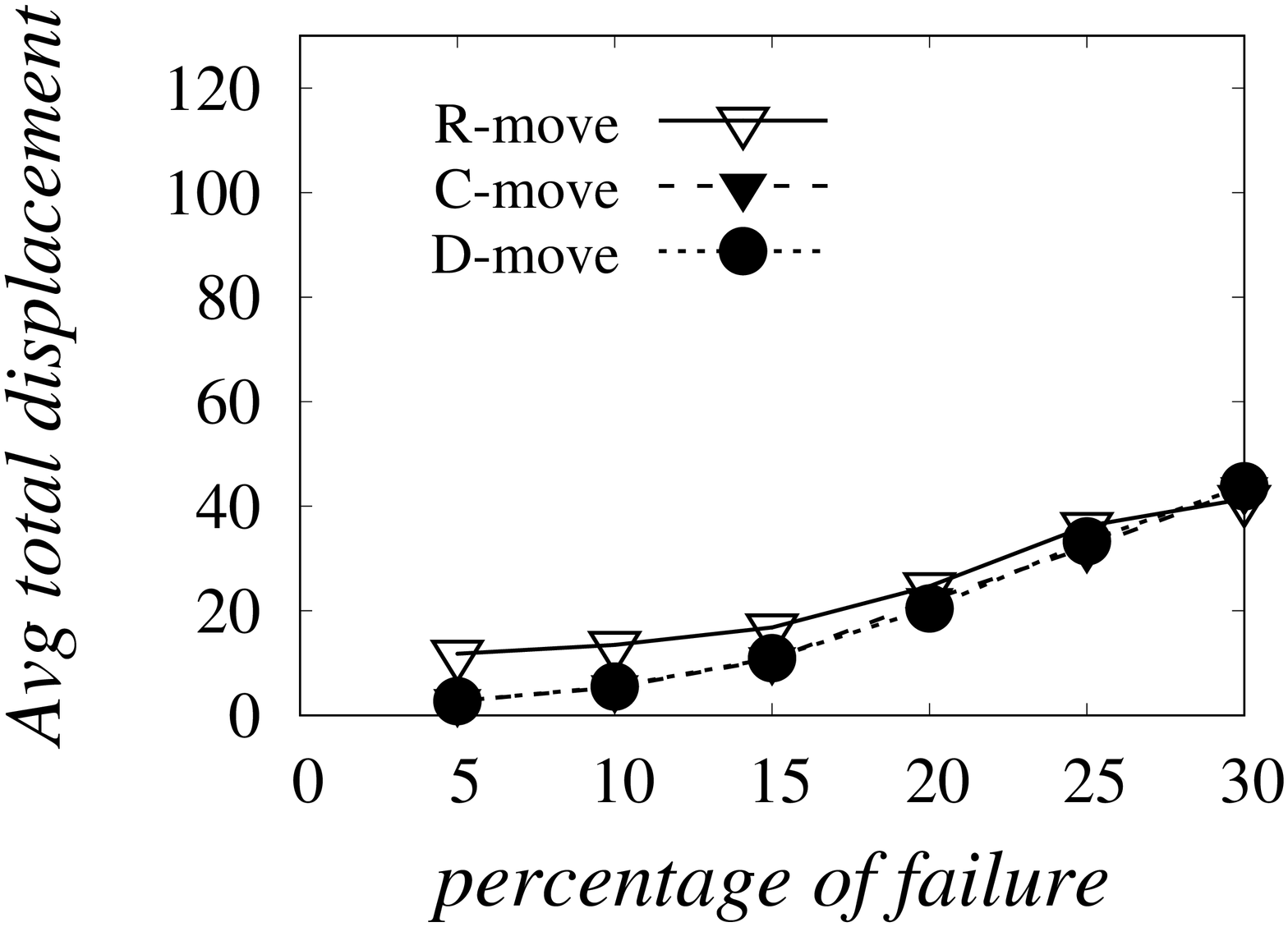}
\caption{ Average total displacement for each failure (a)  N=140    (b)  N=160    (c)  N=180 }
\label{chart:avg_total_disp}
\end{figure}

Figure \ref{chart:avg_total_disp} shows that the average total displacement increases with node failures. But the rate of increase is inversely proportional to the number  of initial sensors deployed. This is expected because with the increase number of sensors the chance of getting closer non barrier recovery node is also increases. We do not show the results for \emph{N-Move} as no sensors are moved in that case. However, as seen earlier, the average recovery rate of \emph{N-Move} is very low with respect to others and hence the zero cost is not of any consequence. Figure \ref{chart:avg_total_disp} shows that the barrier recovery algorithm \emph{C-Move} and \emph{D-Move} make almost same amount of total displacement and is lesser than the \emph{R-Move} algorithm. This is due to alternate path determination stage before shifting as well as preprocessing steps used to find closest recovery node. Most of the cases barrier is recovered only using alternate path without shifting.

%% Remaining energy discussion

\begin{figure}[!h]
	\centering

	\includegraphics[angle=-90, width=8cm]{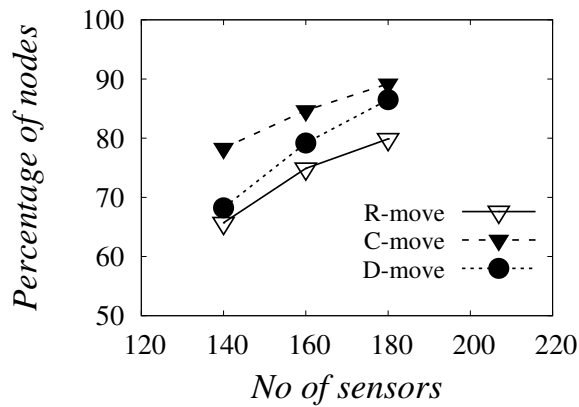}
	\caption{ Remaining energy more than 90\% after recovery}
	\label{chart:avg_energy}
\end{figure}

Figures \ref{chart:avg_energy} presents average energy profile of the sensors after recovery. We count the number of sensors whose remaining energy is more than 90\% of their initial energy level after recovery.  Figures \ref{chart:avg_energy} shows how high-energy profile of the sensors varies with the number of sensors. It is expected that the number of sensors with high remaining energy increases  as sensor number of initial sensors increases. But one can see from the figure that the number of sensor with high remaining energy  in {\em C-Move} algorithm is higher than other two algorithms, and it is constantly high for different values of N. In this result the \emph{N-Move} algorithm is not included, as we consider the energy expenditure due to sensor mobility not for message passing overhead.

%There is also an initial barrier construction cost exist within these two algorithm.
%We show the initial total displacement to construct the first barrier in Figure \ref{chart:initial_disp}.

%\begin{figure}
%\centering
%\includegraphics[angle=-90, width=4cm]{Chapter6/Figure/sigma=9/N120_140_160_avg_initial_disp.eps}
%\caption{Average total initial displacement for N=120,140 and 160}
%\label{chart:initial_disp}
%\end{figure}

\section{Conclusion}
\label{sec:conclusions}

In this paper, we have addressed the problem of maintaining barrier coverage in the presence of sensor failure. After the creation of initial barrier, our proposed algorithm will reconstruct an alternate barrier after every sensor failure. The reconstruction is done either by finding an alternate barrier segment close to the failed sensor node which can be patched up with the broken barrier to reconstruct the whole barrier, or by moving a small number of nodes in a cascaded fashion to recover the gap created by failed node. During node shifting remaining energy and remaining displacement capacity of the sensors are taken into consideration. Barrier repairing by cascaded node shifting with limited mobility sensors problem is formulated as minimum cost bipartite matching problem. A possible localized implementation using message-passing is presented. Our centralized algorithm can handle multiple failures but our localized algorithm is capable of handling multiple failures one after another but not simultaneous failures. Experimental results are presented to show that the algorithm performs well in comparison with existing centralized algorithm which reconstructs the whole barrier by placing existing nodes in equal separation.  In future, we will extend the solution, where instead of restoring the existing barrier with minimum total displacement, we will create a new barrier with minimum total movement and it also respects the remaining energy of the existing sensors.

\remove{
%%%%%%%%%%%%%%
\section{Acknowledgment}

Insert the Acknowledgment text here.
%%%%%%%%%%%%
}

%\section{Appendices}
%
%Appendices are allowed but please be aware that these are included in the overall word count.

\end{document}